\documentclass[a4paper,USenglish,cleveref,autoref,thm-restate]{lipics-v2021}

\hideLIPIcs

\usepackage{stmaryrd}
\usepackage{xspace}
\usepackage{mathtools}
\usepackage{tikz}
\usepackage{pgfplots}
\usepackage{bm}

\pgfplotsset{compat=1.18}

\definecolor{blueLink}{rgb}{0,0.2,0.8}
\hypersetup{colorlinks,linkcolor=blueLink,urlcolor=blueLink,citecolor=blueLink}

\newtheorem{algo}{Algorithm}

\DeclarePairedDelimiter{\ceil}{\lceil}{\rceil}

\newcommand{\ALG}{\textsc{Alg}\xspace}
\newcommand{\OPT}{\textsc{Opt}\xspace}
\newcommand{\OFF}{\textsc{Off}\xspace}
\newcommand{\DET}{\textsc{Det}\xspace}
\newcommand{\GREEDY}{\textsc{Greedy}\xspace}
\newcommand{\E}{\mathbf{E}}
\newcommand{\detlog}{\textsc{Exp}\xspace}

\DeclareFontShape{T1}{lmr}{m}{scit}{<->ssub * lmr/m/scsl}{}
\DeclareFontShape{T1}{lmr}{bx}{sc}{<->ssub * lmr/b/n}{}
\input{t1lmss.fd}
\DeclareFontShape{T1}{lmss}{m}{sc}{<->ssub * lmr/m/sc}{}
\input{Ustmry.fd}
\DeclareFontShape{U}{stmry}{m}{n}{<->stmary10}{}
\DeclareFontShape{U}{stmry}{b}{n}{<->ssub * stmry/m/n}{}

\title{Competitive Transaction Admission in PCNs: Online Knapsack with Positive and Negative Items}
\titlerunning{Competitive Transaction Admission in PCNs}

\author{Marcin Bienkowski}{University of Wrocław, Poland}{marcin.bienkowski@cs.uni.wroc.pl}{}{}
\author{Julien Dallot}{TU Berlin, Germany}{judafa@protonmail.com}{}{}
\author{Dominik Danelski}{TU Berlin, Germany}{dominik@danelski.pl}{}{}
\author{Maciej Pacut}{TU Berlin, Germany}{maciej@inet.tu-berlin.de}{}{}
\author{Stefan Schmid}{TU Berlin and Weizenbaum Institute, Germany}{stefan.schmid@tu-berlin.de}{}{}

\authorrunning{M. Bienkowski, J. Dallot, D. Danelski, M. Pacut and S. Schmid}



\funding{This project was supported by the German Research Foundation (DFG), SPP 2378 (ReNO-2), 2025--2029, and by the Polish National Science Centre grant 2022/45/B/ST6/00559}

\begin{document}
\nolinenumbers

\maketitle

\begin{abstract}
  Payment channel networks (PCNs) are a promising approach to making cryptocurrency transactions faster and more scalable.
  At their core, PCNs bypass the blockchain by routing transactions through intermediary channels.
  However, a channel can forward a transaction only if it has the necessary funds: the problem of keeping the channels balanced is a current bottleneck for the PCN's transaction throughput.

  This paper considers the problem of maximizing the number of transactions accepted by a~channel in a PCN.
  Previous works either considered the associated optimization problem with all transactions known in advance or developed heuristics tested on particular transaction datasets.
  This work, however, considers the problem in its purely online form where the transactions are arbitrary and revealed one after the other.

  We show that the problem can be modeled as a new online knapsack variant where the items (transaction proposals) can be either positive or negative depending on the direction of the transaction.
  The main contribution of this paper is a deterministic online algorithm that is $O(\log B)$-competitive, where $B$ is the knapsack capacity (maximum allowed channel balance).
  We complement this result with an asymptotically matching lower bound of $\Omega(\log B)$ which holds for any randomized algorithm, demonstrating our algorithm's optimality.
\end{abstract}


\section{Introduction}

Blockchain-based decentralized cryptocurrencies suffer from well-documented scalability issues~\cite{DBLP:conf/fc/CromanDEGJKMSSS16,8431962,DBLP:conf/atal/JainSG21} that hinder their adoption for daily payments.
The promising solutions for the blockchain scalability problem are Payment Channel Networks (PCN)~\cite{poon2015bitcoin,Raiden}, which allow users to bypass the blockchain and perform private peer-to-peer transactions.
A~PCN consists of a set of payment channels that allow its two end-users to perform fast, off-chain transactions to one another.
The blockchain is called for an initial funding transaction at a~channel's opening and the two end-users can afterwards freely transact as long as the balance remains within the allocated funds.
In general, two users willing to perform a fast transaction do not share a common channel: a typical transaction is rather routed through the network from channel to channel, each user forwarding the payment along the routing path.

A PCN crucially needs to keep its set of channels balanced to be able to route transactions with high throughput.
An unbalanced channel otherwise becomes unidirectional and requires a~costly call to the blockchain to regain liquidity.
Efficiently balancing the channels to maximize throughput is considered the main bottleneck of PCNs and the subject of intensive research.
Existing works either designed routing heuristics to minimize channel balance~\cite{DBLP:conf/mascots/AwathareSARB21,DBLP:journals/comcom/MercanEA21,DBLP:conf/nsdi/SivaramanVRNYMF20} or considered the associated optimization problem with full knowledge of the incoming requests~\cite{DBLP:conf/wdag/ChatterjeeK0SY25,DBLP:conf/sirocco/SchmidSY23}. 

In this paper, we consider the problem of maximizing the transaction throughput (the number of performed transactions) of a channel in a PCN such as the Lightning network~\cite{poon2015bitcoin}.
As in~\cite{DBLP:conf/sirocco/SchmidSY23}, we solve the problem locally for a single channel.
Unlike previous works, however, we solve the problem in the purely online setting where the incoming transaction proposals are arbitrary and unknown to the algorithm until revealed. As we demonstrate in this paper, commonly-used greedy heuristics that accept every transaction proposal that fits in the channel are inherently suboptimal; see \autoref{fig:greedy} for an illustration.

\begin{figure}[t]
  \centering
  
\usetikzlibrary{patterns, patterns.meta}
\pgfmathsetmacro{\limitlinesleft}{-2}
\pgfmathsetmacro{\onerowheight}{4.3}
\pgfmathsetmacro{\tworowheight}{3.2}
\pgfmathsetmacro{\threerowheight}{2.2}
\pgfmathsetmacro{\fourrowheight}{2.0}
\pgfmathsetmacro{\fiverowheight}{0}
\pgfmathsetmacro{\sixrowheight}{-0.2}
\begin{tikzpicture}[scale=0.8, >=stealth]

\node at (-0.4, \onerowheight/2 + \tworowheight/2) {\footnotesize Transactions};
\node at (-0.2, \tworowheight*0.72 + \threerowheight*0.28) {\footnotesize Accept $\checkmark$};
\node at (-0.4, \tworowheight*0.28 + \threerowheight*0.72) {\footnotesize or Reject \bm{$\times$}};
\node at (-0.5, \fourrowheight*0.68 + \sixrowheight*0.32) {\footnotesize Balance};
\node at (-0.5, \fourrowheight*0.46 + \sixrowheight*0.54) {\footnotesize ($B = 10$)};

\draw (\limitlinesleft, \sixrowheight) -- (\limitlinesleft, \onerowheight);
\draw (1, \sixrowheight) -- (1, \onerowheight);
\foreach \x in {2,...,8} {
  \draw (\x, 2.2) -- (\x, \onerowheight);
}
\draw (9, \sixrowheight) -- (9, \onerowheight);

\draw (\limitlinesleft, \onerowheight) -- (9, \onerowheight); 
\draw (\limitlinesleft, \tworowheight) -- (9, \tworowheight); 
\draw (\limitlinesleft, \threerowheight) -- (9, \threerowheight); 
\draw (1, \fourrowheight) -- (9, \fourrowheight); 
\draw (1, \fiverowheight) -- (9, \fiverowheight); 
\draw (\limitlinesleft, \sixrowheight) -- (9, \sixrowheight); 

\node at (1.5, \onerowheight/2 + \tworowheight/2) {\footnotesize \bm{$+3$}};
\node at (2.5, \onerowheight/2 + \tworowheight/2) {\footnotesize \bm{$-2$}};
\node at (3.5, \onerowheight/2 + \tworowheight/2) {\footnotesize \bm{$-5$}};
\node at (4.5, \onerowheight/2 + \tworowheight/2) {\footnotesize \bm{$+14$}};
\node at (5.5, \onerowheight/2 + \tworowheight/2) {\footnotesize \bm{$+1$}};
\node at (6.5, \onerowheight/2 + \tworowheight/2) {\footnotesize \bm{$+1$}};
\node at (7.5, \onerowheight/2 + \tworowheight/2) {\footnotesize \bm{$+1$}};
\node at (8.5, \onerowheight/2 + \tworowheight/2) {\footnotesize \bm{$+1$}};

\node at (1.5, 2.75) {$\checkmark$};
\node at (2.5, 2.75) {$\checkmark$};
\node at (3.5, 2.75) {$\checkmark$};
\node at (4.5, 2.75) {$\checkmark$};
\node at (5.5, 2.75) {\bm{$\times$}};
\node at (6.5, 2.75) {\bm{$\times$}};
\node at (7.5, 2.75) {\bm{$\times$}};
\node at (8.5, 2.75) {\bm{$\times$}};

\path[pattern={Lines[angle=50, yshift=.5pt]}, pattern color=black] (1, \fourrowheight) rectangle (9, \fourrowheight+0.2);
\path[pattern={Lines[angle=-50, yshift=.5pt]}, pattern color=black] (1, \fiverowheight) rectangle (9, \fiverowheight-0.2);

\pgfmathsetmacro{\maxheight}{1.92}
\draw[dashed] (1, 1) -- (9, 1); 
\draw[->, very thick] (1,1.0) -- (2,1.3);
\draw[->, very thick] (2,1.3) -- (3,1.1);
\draw[->, very thick] (3,1.1) -- (4,0.6);
\draw[->, very thick] (4,0.6) -- (5,\maxheight);
\draw[->, very thick] (5,\maxheight) -- (6,\maxheight);
\draw[->, very thick] (6,\maxheight) -- (7,\maxheight);
\draw[->, very thick] (7,\maxheight) -- (8,\maxheight);
\draw[->, very thick] (8,\maxheight) -- (9,\maxheight);

\end{tikzpicture}
\caption{
  A run of the algorithm \GREEDY on a single channel, for the sequence of transaction proposals shown in the first row.
  The second row gives the decision of \GREEDY for each transaction proposal, and the third row plots the channel balance after each decision.
  Assuming the sequence stops after eight transaction proposals, \GREEDY is suboptimal as it accepts the fourth proposal of size $+14$, which prevents it from accepting the last four transaction proposals of size $+1$.
  The channel balance has to be between $-B = -10$ and $B = 10$.
}
\label{fig:greedy}
\end{figure}

In contrast, we present a deterministic online algorithm for this problem and show that it is $O(\log B)$-competitive, where $B$ is the maximum allowed balance of the single channel divided by two.
We then show a~lower bound of $\Omega(\log B)$ for all algorithms, including randomized ones, making our algorithm asymptotically optimal.

We cast the problem as a new online knapsack variant (see, e.g., \cite{Lueker95,DBLP:conf/icalp/Marchetti-SpaccamelaV87}) where the items (transaction proposals) can have positive or negative sizes (transaction amounts) depending on the direction of the transaction.
The knapsack capacity corresponds to the maximum allowed balance of the single channel.
Using signed items in knapsack problems is new.
If we omit negative items, however, our model is similar to that of Zhou, Chakrabarty and Lukose~\cite{chakrabarty2008online} in that all our items have an (absolute) size within a known interval (as in the Lightning protocol, where channels broadcast their minimum and maximum accepted transaction amounts). However, the presence of negative items and potentially unbounded input length makes the results of~\cite{chakrabarty2008online} inapplicable to our setting.


\section{Model and Notations}
\label{sec:model}

\subparagraph*{Knapsack and its Relation to PCN}
We consider an online version of the knapsack problem with items that can have both positive and negative sizes.
Real-valued item sizes $\sigma_1, \sigma_2, \dots$ arrive one by one in an online fashion, and the task is to decide to either accept or reject each item when it arrives without knowledge of future items.
We assume that the absolute value of any item size is from interval $[1, m]$, where $m$ is a value known to an online algorithm.

Let $B$ denote the \emph{capacity} of the knapsack.
At any time, we define the \emph{state} of the knapsack as a value $s \in \mathbb{R}$ which evolves as items are accepted.
Initially, the knapsack's state is $s_0 \in [-B, B]$; unless stated otherwise, we assume that $s_0 = 0$.
Upon arrival of an~item~$\sigma_i$, the knapsack's state updates as $s \leftarrow s + \sigma_i$ in case the algorithm accepted the item, or stays unchanged otherwise.

The goal is to maximize the number of accepted items while maintaining the knapsack's state within the interval $[-B, B]$.

\subparagraph*{Modeling Payment Channel Networks}
Our knapsack model reflects the case of a single channel in a Payment Channel Network (PCN) such as the Lightning~\cite{poon2015bitcoin} or Raiden~\cite{Raiden} networks.
In a PCN, two users can establish a \emph{payment channel} to perform transactions without the need for on-chain confirmations.
The two end-users may perform transactions freely through the channel as long as the remaining allocated funds both remain positive. 
Those transactions may be for private purposes between the two end-users; however, the majority of the transactions come from outside of the channel.
In that case, the transaction represents one hop as part of a routing path between distant users of the network.
Over time, the PCN computes routing paths which may include the considered channel, and when this happens, the channel can either accept or reject the transaction proposal.

In our model, the sequence of items represents the transaction proposals coming from the rest of the network.
The restriction on the absolute sizes of the items in $[1, m]$ models the minimum and maximum accepted transaction sizes that the channels individually decide and broadcast to the PCN (e.g., as done by the Lightning protocol~\cite{poon2015bitcoin}).
Note that the minimum size of $1$ implies no loss of generality as the items, the knapsack capacity and our algorithm (its acceptance curve) can be easily re-scaled to fit any positive minimum size: the parameter~$m$ can then be interpreted as the upper bound on the ratio between maximum and minimum items sizes.
Note that this minimum size cannot be arbitrarily small as cryptocurrencies typically come with a smallest possible transaction amount (the Satoshi for Bitcoin and therefore the Lightning network): the channels in our dataset usually set the minimum accepted item to $1$~Satoshi.
Finally, our model handles the case when both end-users do not have the same initial funds by setting a non-zero initial balance $s_0$.

\subparagraph*{Competitive analysis}
We explain here the competitive analysis framework in its maximization form; this measure is discussed at length in the book of Borodin and El-Yaniv~\cite{DBLP:books/daglib/0097013}.
For each input $\sigma$, let $\ALG(\sigma)$ denote the total gain of the online algorithm (the number of accepted items) and $\textsc{Opt}(\sigma)$ denote the total gain of an optimal offline solution.
A~deterministic online algorithm is said to be $c$-competitive if there exists a parameter $\beta$, such that for any item sequence $\sigma$ it holds that 
\begin{equation}
    \label{eq:comp_def}
    c \cdot \ALG(\sigma) \geq \OPT(\sigma) - \beta.
\end{equation}
The parameter $\beta$ can depend on other model parameters such as $B$ or $m$, but cannot depend on input $\sigma$.
For randomized algorithms, in the definition above, we replace $\ALG(\sigma)$ by its expectation (taken over random choices of the algorithm).
A competitive ratio is said to be \emph{strict} when $\beta=0$.


\section{Contributions and Technical Novelty}
\label{sec:contributions}

Our main contribution is an asymptotically optimal competitive ratio for the transaction admission problem.
Our algorithm has a~competitive ratio of $O(\log B)$ (\autoref{thm:det-logB}), where $B$ is the knapsack capacity (channel balance has to be between $-B$ and $B$).
For this guarantee to hold, we require a mild technical assumption on item sizes: all items must have an absolute size in $[1, m]$ where $m \le b$ and $b = B / \ln B$.
The algorithm is simple, efficient, and uses no memory: the algorithm accepts or rejects each incoming proposal solely by comparing its amount to a carefully chosen threshold depending on the current balance (see \autoref{alg:det_logB} and its illustration in \autoref{fig:exp_curve}).

We complement our result with a lower bound of $\Omega(\log B)$ (\autoref{thm:rand_lb}), proving the asymptotic optimality of our algorithm.
Note that our lower bound is formulated as $\Omega(\log m)$ where $m$ is the largest absolute size of an~item; this bound is matching, for instance, when  $m = b = B / \ln B$ (as allowed by our algorithm).
This lower bound applies even to randomized algorithms, making our (deterministic) algorithm optimal across the standard frameworks of competitive analysis.

From a theoretical perspective, our results improve on the state-of-the-art of the classic online knapsack by generalizing it to signed item sizes.
This generalization deeply affects the problem as its input length becomes inherently unbounded, which requires a new set of tools.

In addition, we show a worst-case input sequence for the widely-used greedy algorithm~\cite{DBLP:journals/comcom/MercanEA21,DBLP:journals/corr/abs-2109-11665} which shows an exponential gap between the competitive ratio of our algorithm and that of greedy~(\autoref{thm:greedy_lineary_cr}).

Finally, our algorithm shows very good performance in simulations (see~\autoref{sec:evaluations}).
We show that its performance is equivalent to the commonly-used greedy algorithm, both on a~single, isolated channel, and in snapshot topologies of the Lightning network against daily, non-adversarial transactions.
Our theoretical results, however, ensure optimal performance guarantees even in adversarial environments, where the greedy algorithm performs poorly.


\section{Lower Bound on Greedy}
\label{sec:lb}

\GREEDY~\cite{DBLP:journals/comcom/MercanEA21,DBLP:journals/corr/abs-2109-11665} is a widely-used algorithm that accepts every incoming item that fits in the knapsack.
This section establishes that the greedy algorithm has a~competitive ratio of~$\Omega(m)$.
Recall that $B$ is the size of the knapsack, that is, the knapsack's state must stay in the interval $[-B, B]$, and $m \in [1, B]$ is the upper bound on the absolute item size.

\begin{theorem}
  \label{thm:greedy_lineary_cr}
  The competitive ratio of \GREEDY is $\Omega(m)$.
\end{theorem}

\begin{proof}
    We start with an initial knapsack's state $s_0 = 0$.
    Let $d = \ceil{B/m}$ and $m' = B/d$. Note that $m' \in [1,m]$ and $m' = \Omega(m)$. 
    We will show a lower bound of $m'$ on the competitive ratio of \GREEDY.
    
    We construct an input sequence $\sigma$ in successive phases.
    \begin{itemize}
    \item
        \textbf{Phase 0.}
        Issue $d$ items of size $m'$ followed by a sequence of $B$ items of size $1$.
    \item
        \textbf{Odd phases.}
        Issue $2d$ items of size $-m'$ followed by a sequence of $2B$ items of size~$-1$.
    \item
        \textbf{Even phases.}
        Issue $2d$ items of size $m'$ followed by a sequence of $2B$ items of size $1$.
    \end{itemize}
  \GREEDY accepts the first $d$ items of phase $0$ and accepts the first $2d$ items of each phase afterwards.
  Now consider an (offline) algorithm \OFF that rejects the items accepted by \GREEDY and vice versa.
  \OFF accepts $B$ items during phase $0$ and $2B$ items in each phase afterwards.
  By summing up over all phases of $\sigma$, we obtain $\GREEDY(\sigma) = (B/d) \cdot \OFF(\sigma) =  m' \cdot \OFF(\sigma) \leq m' \cdot \OPT(\sigma)$.

  By increasing the number of phases, we can ensure that $\OPT(\sigma)$ is arbitrarily large. This implies that the lower bound of $m'$ holds even if we allow an arbitrary additive constant in the competitive ratio definition~\eqref{eq:comp_def}.
\end{proof}


\section{Deterministic Online Algorithm}
\label{sec:det}

In this section, we present an online deterministic algorithm that is $O(\log B)$-competitive when $B \geq 4.1$ and $m \leq b$ (where $b = B/\log B$).

\begin{algo}[\detlog]
  \label{alg:det_logB}
  We define the function $f(x) = b \cdot \exp(-|x|/b)$.
  Let $\sigma_i$ be the incoming item and $s$ be the current knapsack state.
  Accept $\sigma_i$ if $\sigma_i$ and $s$ have opposite signs or~$|\sigma_i| \le f(s)$.
\end{algo}

\begin{figure}[t]
\centering
\begin{tikzpicture}
  \begin{axis}[
      axis lines = middle,
      xlabel = \(s\),
      ylabel = \(f(s)\),
      xmin = -105,
      xmax = 105,
      ymin = 0,
      ymax = 23,
      samples = 300,
      smooth,
      scale only axis=true,
      width=8.7cm,
      height=5cm,
      grid = both,
      legend pos = north west,
    ]
    \addplot[blue, thick, domain=-105:105] {((100 - 1)/ln(100)) * exp(-abs(x) * ln(100)/(100 - 1))};
    \addlegendentry{\( f(s) = b \cdot e^{- |s| / b}\)}
  \end{axis}
\end{tikzpicture}
\caption{
  The acceptance curve of the \detlog algorithm for $B = 100$.
  The x-axis is the knapsack state, the y-axis is the absolute size of an item.
  \detlog accepts an item if either that item and the knapsack's state have different signs or else if the corresponding point is below the curve.
}
\label{fig:exp_curve}
\end{figure}

\subparagraph*{Correctness}
We first discuss the correctness of \autoref{alg:det_logB}.

\begin{theorem}
  \label{thm:correctness}
  \autoref{alg:det_logB} maintains the knapsack's state within $[-B, B]$ when $B \ge 4.1$.
\end{theorem}

\begin{proof}
  Let $s \in [0, B]$ be a non-negative knapsack's state; the non-positive case is symmetric.
  Let $x$ be the item that \detlog accepts. We will show that $s + x \in [-B,B]$.
  
  If $x < 0$, the new state is $s + x \geq s - m \geq s - b \geq s - B \geq -B$. 
  
  If $x > 0$, we define $\hat{s} = b \cdot \ln b$. Note that $f(\hat{s}) = 1$. Hence, we must have $s \leq \hat{s}$ as otherwise any positive item (of size at least $1$) would be rejected by \detlog.
  The new state is then
  \begin{align*}
    s + x &\le s + f(s) 
        && \text{(by the definition of \autoref{alg:det_logB})}\\
    & \le \hat{s} + f(\hat{s}) 
        && \text{(as $x \mapsto x + f(x)$ is increasing on $[0, B]$)}\\ 
    & = b \cdot \ln b + 1 \\
    & = B - B \cdot \ln \ln B / \ln B  + 1.
    \end{align*}
    For $B \geq 4.1$, the last quantity is at most $B$, concluding the proof.
\end{proof}

\subparagraph*{Competitive ratio}
Now that we have established correctness, we prove the performance guarantees of \autoref{alg:det_logB}.
We start with the following technical bound.

\begin{lemma}
    \label{lem:f_diff_bound}
    Suppose $\detlog$ is in a state $s \geq 0$ and accepts 
    item of size $x > 0$. Then, 
    \[
        \frac{1}{f(s+x)} - \frac{1}{f(s)} \le \frac{e-1}{b}
    \]
\end{lemma}

\begin{proof}
    Recall that by our assumption $x \leq m \leq b$. Thus, we may use the inequality $e^y \le 1 + (e-1) \cdot y$ holding for any $y \in [0, 1]$,
    obtaining
    \begin{align*}
        \exp(x/b) - 1 
        & \le (e-1) \cdot (x/b) \\
        & \le (e-1) \cdot (f(s)/b) 
            && \text{(as \detlog accepts $x$ in state $s$)}\\
        & = (e-1) \cdot \exp(-s/b).
    \end{align*} 
    We can now estimate the desired quantity as
    \begin{align*}
        \frac{1}{f(s+x)} - \frac{1}{f(s)} 
        &= \frac{1}{b} \cdot \left(\exp\left(\frac{s+x}{b}\right) - \exp\left(\frac{s}{b}\right)\right) \\
        &= \frac{1}{b} \cdot \exp\left(\frac{s}{b}\right) \cdot \left(\exp\left(\frac{x}{b}\right) - 1\right) \\
        &\leq \frac{e-1}{b}.
        \qedhere
    \end{align*}
\end{proof}

\begin{theorem}
  \label{thm:det-logB}
  \detlog is $O(\log B)$-competitive.
\end{theorem}

\begin{proof}
  Let $\sigma = \sigma_1, \sigma_2, \dots$ be an input sequence.
  We will prove the claim with the potential function technique.
  Let \OPT be an optimal offline algorithm.
  Let $s_i$ and $s^*_i$ be the states of \detlog{} and \OPT, respectively, after they handle item $\sigma_i$.

    Let $d_i = 2B - s^*_i$ if $s_i \ge 0$ and $d_i = 2B + s^*_i$ otherwise. Intuitively, 
    $d_i$ equals $B$ plus the distance from \OPT's state to that knapsack boundary which is nearest to \detlog's state.
      
  We define the potential function $\Phi$ after serving item $\sigma_i$ as 
  \begin{align*}
    \Phi(i) = d_i \,/\, f(s_i).
  \end{align*}
    Note that $\Phi(i) \geq 0$. Moreover, for any $i$, we have $|d_i| = O(B)$ and  $f(s_i) = b \cdot \exp(-|s_i|/b) \ge b \cdot \exp(-B/b) = 1/\ln B$, and thus $\Phi(i) = O(B \cdot \log B)$.

  For each $i$, we use $\detlog(i)$ and $\OPT(i)$ to denote the profits of \detlog and \OPT on item $\sigma_i$, respectively. 
  Our goal now is to show that the following inequality holds for each $i$.
  \begin{align}
    \label{ineq:potential}
    \OPT(i) + \Phi(i) - \Phi(i-1) \leq (1 + (5e - 3) \cdot \ln B) \cdot \detlog(i).
  \end{align}
    
  To show it, we first observe that flipping the signs of $\sigma_i$, $s_{i-1}$ and $s_{i-1}^{*}$ leaves $\Phi(i-1)$ and $\Phi(i)$ unchanged and keeps the same acceptance decisions for \detlog regarding $\sigma_i$.
  We can therefore assume, without loss of generality, that $s_{i-1} \geq 0$.
  For succinctness, we use $\Delta_i \Phi = \Phi(i) - \Phi(i-1)$.
  We distinguish between two cases depending on whether \detlog accepted or rejected~$\sigma_i$.
  
  \begin{enumerate}
    \item \detlog rejects $\sigma_i$.

    If \OPT also rejects $\sigma_i$, the inequality \eqref{ineq:potential} trivially holds as both sides are zero. Thus, in the remaining part of this case, we assume that \OPT accepts $\sigma_i$.

    Recall that $s_{i-1} \geq 0$. 
    The item rejection by $\detlog$ is possible only if $\sigma_i > f(s_{i-1})$. 
    In the considered case, $s_i = s_{i-1}$. As \OPT accepts $\sigma_i$, we have $s^*_i = s^*_{i-1} + \sigma_i$,
    and thus $d_i = d_{i-1} - \sigma_i$. 
    Using these relations, we can estimate $\Delta_i \Phi$ as
    \begin{align*}
        \Delta_i \Phi 
        = \frac{d_i}{f(s_i)} - \frac{d_{i-1}}{f(s_{i-1})} 
        = \frac{d_{i-1} - \sigma_i}{f(s_i)} - \frac{d_{i-1}}{f(s_{i-1})} 
        = - \frac{\sigma_i}{f(s_{i-1})} 
        < -1.
    \end{align*}
    As $\OPT(i) = 1$ and $\detlog(i) = 0$, \eqref{ineq:potential} follows.

    \item \detlog accepts $\sigma_i$.
    
    In this case $\detlog(i) = 1$ and $\OPT(i) \leq 1$. Hence, for showing \eqref{ineq:potential}, 
    it suffices to show that $\Delta_i \Phi \le (5e - 3) \cdot \ln B$.
    Using $s_i = s_{i-1} + \sigma_i$, we obtain
    \begin{align*}
        \Delta_i \Phi 
        & = \frac{d_i}{f(s_i)} - \frac{d_{i-1}}{f(s_{i-1})} 
        = \left(\frac{d_i}{f(s_i)} - \frac{d_i}{f(s_{i-1})}\right) + \left(\frac{d_{i}}{f(s_{i-1})} - \frac{d_{i-1}}{f(s_{i-1})}\right) \\
        & \leq d_i \cdot \left(\frac{1}{f(s_{i-1}+\sigma_i)} - \frac{1}{f(s_{i-1})}\right) 
            + \frac{|d_{i} - d_{i-1}|}{f(s_{i-1})} 
    \end{align*}
    We denote the two summands in the line above as $\Delta_i^s \Phi$ and $\Delta_i^d \Phi$, respectively. 
    We will show that $\Delta_i^s \Phi \le (3e - 3) \cdot \ln B$ and $\Delta_i^d \Phi \le 2 e \cdot \ln B$, which will conclude the proof of this case.
    \begin{itemize}
    \item First, we upper-bound $\Delta_i^s \Phi$. If $\sigma_i < 0$, then 
    $\Delta_i^s \Phi \le 0$. Otherwise $\sigma_i > 0$, and we can use \autoref{lem:f_diff_bound} to obtain
    \begin{align*}
      \Delta^s_i \Phi 
        & \leq d_i \cdot \frac{e-1}{b} \\
        &\le 3 \cdot (e-1) \cdot \ln B.
            && \text{(as $d_i \le 2 \cdot B + s^*_i \leq 3 \cdot B$)} 
    \end{align*}
    \item Next, we upper-bound $\Delta_i^d \Phi$. 
        Note that the state of \OPT changes at most by $\sigma_i$, and thus $|s^*_i - s^*_{i-1}| \leq \sigma_i$.
        We consider two sub-cases, depending on whether $\detlog$ changes the sign of its state or not. 
        
        \begin{itemize}
        \item 
        If $\detlog$ does not change the sign of its state, then $|d_i - d_{i-1}| = |s^*_i - s^*_{i-1}| \leq \sigma_i$.
        As $\detlog$ accepts $\sigma_i$, we have $\sigma_i \le f(s_{i-1})$, and therefore 
        $\Delta_i^d \Phi \le 1$.
        \item
        If $\detlog$ changes the sign of its state, then $|d_i - d_{i-1}| \leq |s^*_i| + |s^*_{i-1}| \leq 2 B$. However, as $\sigma_i \leq m \leq b$, the sign change is possible only if $|s_{i-1}| \leq b$, in which case $f(s_{i-1}) = b \cdot \exp(- |s_{i-1}| / b) \geq b \cdot e^{-1}$. 
        Combining the above bounds, we get $\Phi_i^d  \le 2B / (b \cdot e^{-1}) = 2 e \cdot \ln B$
        \end{itemize}
        
        Putting everything together, we have shown that $\Delta_i \Phi \leq \Delta_i^s \Phi 
        + \Delta_i^d \le (5e - 3) \cdot \ln B$, which concludes the proof of \eqref{ineq:potential} of this case.
    \end{itemize}
      
  \end{enumerate}

  We proved that \eqref{ineq:potential} holds for each request $\sigma_i$.
  By summing this relation over all requests in the input sequence, we obtain that 
  \begin{align*}
        (1 + (5e - 3) \cdot \ln B) \cdot \detlog(\sigma)
        & \ge \OPT(\sigma) + \Phi(|\sigma|) - \Phi(0) \\
        & \ge \OPT(\sigma) - O(B \cdot \log B),
  \end{align*}
  where the last inequality follows as $\Phi(i) \geq 0$ and $\Phi(i) = O(B \log B)$ for any $i$. 
  This shows that the competitive ratio of \detlog is at most 
  $(1 + (5e - 3) \cdot \ln B) = O(\log B)$.
\end{proof}


\section{Randomized Lower Bound}
\label{sub:lb_logB}

In this section, we present a lower bound of $\Omega(\log m)$ on the competitive ratio of any randomized algorithm; recall that $m$ is the maximum item size (in absolute value).

We may assume that $m \geq 4$, as otherwise the lower bound follows trivially. Let $q = \lfloor \log_2 \,(m/2) \rfloor$ and $h = \lceil 2 B / 2^q \rceil$. Note that $q \geq 1$ and $2^q \leq m/ 2 \leq B$, and thus $h \geq 2$.

To show the lower bound, we will use Yao's minimax principle~\cite{Yao77,DBLP:books/daglib/0097013}. To this end, we show how to construct, for any positive integer $n$, a probability distribution $\pi_n$ over inputs, such that
\begin{itemize}
  \item $\lim_{n \to \infty } \OPT(\pi_n) = \infty$ and 
  \item $\E_{\sigma \sim \pi_n}[\OPT(\sigma)] / \E_{\sigma \sim \pi_n}[\DET(\sigma)] \geq q / 8$ for every deterministic algorithm \DET.
\end{itemize}
Then, the minimax principle implies that the competitive ratio of every \emph{randomized} algorithm is at least $q/8 = \Omega(\log m)$.

\subparagraph*{Random inputs}
For an integer $z \in [1, q]$, we define the following input subsequence, called a~\emph{positive phase of duration $z$}. It consists of $2 \cdot h$ items of size $2^{q-1}$, followed by $4 \cdot h$ items of size $2^{q-2}$, then by $8 \cdot h$ items of size $2^{q-3}$, and so on, up to $2^{z} \cdot h$ items of size $2^{q-z}$. A~\emph{negative phase of duration $z$} is defined analogously, but with negative item sizes.

The probability distribution $\pi_n$ is defined implicitly by the following random process of generating input sequences. For an integer $t \in \{1, \ldots, 2 \cdot n\}$, let $z_t$ be an independently drawn random variable, such that for every $i \in \{1, \ldots, q\}$,
\begin{equation}
  z_t = i \quad \text{with probability $2^{-i} \cdot \frac{1}{1-2^{-q}}$}.
\end{equation}
The random input consists of $2 \cdot n$ phases, where the duration of the $t$-th phase is equal to $z_t$. Moreover, odd-numbered phases are positive and even-numbered ones are negative.

\subparagraph*{Estimating the expected gain of OPT}
We start with the following key observation about phases.

\begin{lemma}
  \label{lem:phase_gain}
  Suppose an algorithm state is at most $0$ before a positive phase of duration~$z$. It is possible to collect at least $B \cdot 2^{z-q-1}$ items in this phase, and end in a state of at least~$B/2$.
\end{lemma}

\begin{proof}
  Observe that the total size of all items of size $2^{q-z}$ in the phase is $2^z \cdot h \cdot 2^{q - z} = 2^q \cdot h = \lceil 2B / 2^q \rceil \cdot 2^q \geq 2B$. 

  Suppose an algorithm collects as many of these items as possible. Independently of its starting state, it ends the phase in a state greater than $B - 2^{q-z} \geq B - 2^{q-1} \geq B/2$. As it starts in a state of at most $0$, the total size of the collected items is at least $B/2$, and thus the number of collected items is at least $(B/2) / 2^{q-z} = B \cdot 2^{z-q-1}$.
\end{proof}

\begin{lemma}
  \label{lem:lb_off_bound}
  Fix an integer $n$. It holds that
  $\E_{\sigma \sim \pi_n}[\OPT(\sigma)] \geq B \cdot 2^{-q} \cdot q \cdot n$.
\end{lemma}

\begin{proof}
  We consider an offline algorithm $\OFF$ that in the $t$-th phase accepts as many items of absolute size $2^{q - z_t}$ as possible. The state of \OFF is $0$ before the first phase, and thus satisfies the condition of \autoref{lem:phase_gain}. By the iterative application of \autoref{lem:phase_gain} (and its symmetric counterpart for negative phases), the state of \OFF is at most $0$ before each positive phase, and at least $0$ before each negative one. Thus, again by \autoref{lem:phase_gain}, in each phase $t$, \OFF collects at least $B \cdot 2^{z_t-q-1}$ items.
  Since $\OPT(\sigma) \geq \OFF(\sigma)$, we have
  \begin{align*}
    \OPT(\sigma) \geq B \cdot 2^{-q-1} \cdot \sum_{t=1}^{2n} 2^{z_t}.
  \end{align*}
  Using the distribution of $z_t$ yields $\E[2^{z_t}] = \sum_{i=1}^{q} \Pr[z_t = i] \cdot 2^{i} = q / (1-2^{-q}) > q$,
  and thus
  \begin{align*}
    \E_{\sigma \sim \pi_n} [\OPT(\sigma)] 
    &\geq (B \cdot 2^{-q-1}) \cdot \sum_{t=1}^{2n} \E[2^{z_t}]
    > B \cdot 2^{-q} \cdot q \cdot n.\qedhere
  \end{align*}
\end{proof}

\subparagraph*{Estimating the expected gain of an online algorithm}
In a phase of duration $z$, the last batch of items (those of size $2^{q-z}$) constitutes more than half of the items in that phase, and \OPT can focus solely on collecting these items. An online algorithm does not know the phase duration, and it cannot simulate such a strategy.

For an input $\sigma$ chosen according to the above random process, we use~$\sigma^t$ to denote the $t$-th phase of $\sigma$. 

\begin{lemma}
  Fix $t \in \{1, \dots, 2n\}$. For every starting state and every deterministic strategy \DET for serving the phase $\sigma^t$, it holds that $\E[\DET(\sigma^t)] \leq B \cdot 2^{-q+1}$.
  \label{lem:lb_det_bound}
\end{lemma}

\begin{proof}
  Without loss of generality, we assume that phase $t$ is positive (the negative case is symmetric). Recall that the phase duration $z_t$ is an integer between $1$ and $q$, and for each realization of $z_t$, $\sigma^t$ is a prefix of the phase of duration $q$. 

  Let $x^t_1, x^t_2, \ldots, x^t_q$ be the number of items of sizes $2^{q-1}, 2^{q-2}, \ldots, 2^{q-q}$, respectively, that would be accepted by \DET when $z_t = q$. Then, on the actual phase of duration $z_t$, \DET accepts $x_j$ items of size $2^{q-j}$ for $j = 1, \ldots, z_t$, i.e., $\DET(\sigma^t) = \sum_{j=1}^{z_t} x^t_j$. (Note that we do not assume that $x^t_i$ are the same for every phase $t$). We have 
  \begin{equation}
    \sum_{i=1}^{q} x^t_i \cdot 2^{q-i} \leq 2 B. 
    \label{eq:valid_x_choice}
  \end{equation}
  as otherwise \DET would violate the knapsack capacity when $z_t = q$. This condition holds independently of the knapsack's state of \DET at the beginning of phase $t$.
  
  Using the distribution of $z_t$ yields 
  \begin{align*}
    \E[\DET(\sigma^t)] 
    & = \sum_{i=1}^{q} \Pr[z_t = i] \cdot \sum_{j=1}^{i} x^t_j
    = \sum_{j=1}^{q} x^t_j \cdot \sum_{i=j}^{q} \Pr[z_t = i] \\
    & = \sum_{j=1}^{q} x^t_j \cdot \frac{2^{-j} - 2^{-q}}{1-2^{-q}} 
    < \sum_{j=1}^{q} x^t_j \cdot 2^{-j} 
    \leq 2 B \cdot 2^{-q},
  \end{align*}
  where the last inequality follows from \eqref{eq:valid_x_choice}.
\end{proof}

\begin{theorem}
  \label{thm:rand_lb}
  The competitive ratio of any randomized algorithm is $\Omega(\log m)$.
\end{theorem}

\begin{proof}
  Construct the distribution $\pi_n$ as above and fix any deterministic algorithm \DET. 
  By summing the bounds of \autoref{lem:lb_det_bound} over all $2 n$ phases, we obtain that
  \begin{align*}
    \E_{\sigma \sim \pi_n}[\DET(\sigma)] 
    &= \sum_{t=1}^{2 n} \E[\DET(\sigma^t)]
    \leq B \cdot 2^{-q+2} \cdot n, 
  \end{align*}
  and thus using \autoref{lem:lb_off_bound}, 
  \begin{align*}
    \E_{\sigma \sim \pi_n}[\OPT(\sigma)]/\E_{\sigma \sim \pi_n}[\DET(\sigma)] \geq q / 4.
  \end{align*}
  Thus, by the application of the minimax principle~\cite{Yao77,DBLP:books/daglib/0097013}, the competitive ratio of any randomized algorithm is at least $q/4 = \Omega(\log m)$.

  Note that $\lim_{n \to \infty} \E_{\sigma \sim \pi_n}[\OPT(\sigma)] = \infty$ by \autoref{lem:lb_off_bound}, and thus the bound holds even if we allow an arbitrary additive constant in the competitive ratio definition~\eqref{eq:comp_def}.
\end{proof}


\section{Related Work}
\label{sec:related}
This section reviews relevant related works on channel balancing in PCNs and online knapsack problems.
To our knowledge, the online knapsack problem with negative items has never been studied.

\subparagraph*{Channel balancing in PCNs}
Payment channel networks such as the Lightning or Raiden Networks~\cite{poon2015bitcoin,Raiden} enable fast, off-chain transactions between users and have emerged as a~promising scaling solution for cryptocurrencies~\cite{DBLP:conf/nsdi/SivaramanVRNYMF20}.
These networks come with a distinctive set of algorithmic challenges that differ from traditional networking problems.
In particular, the challenges include incentivizing charging mechanisms~\cite{DBLP:conf/middleware/EngelmannKKGW17}, routing for re-balance maximization~\cite{DBLP:conf/mascots/AwathareSARB21,DBLP:conf/wdag/ChatterjeeK0SY25} or channel depletion reduction~\cite{DBLP:journals/comcom/MercanEA21,DBLP:conf/nsdi/SivaramanVRNYMF20}, and privacy concerns~\cite{DBLP:journals/corr/abs-2109-11665}.

The work closest to ours is by Schmid, Svoboda and Yeo~\cite{DBLP:conf/sirocco/SchmidSY23}, who studied the offline problem of maximizing a user's revenue in a payment channel network.
The user's possible actions are the same as ours (accept or reject an incoming transaction), but the objective is to maximize the user's revenue. That is, an algorithm needs to strike a trade-off between charging fees (linear with the transaction size) and the channel balance which may prevent the user from accepting future transactions.
Their offline approximation scheme has the ratio arbitrarily close to $1 + \sqrt{3}$.

\subparagraph*{Online knapsack in unit-profit model}
The online variant of the knapsack problem is a classic optimization problem that has been extensively studied for decades, with the first results dating back to the 1980s. The seminal papers by Marchetti-Spaccamela and Vercellis~\cite{DBLP:conf/icalp/Marchetti-SpaccamelaV87} and Lueker~\cite{Lueker95} observed that no online deterministic algorithm can achieve a bounded competitive ratio, and thus focused on the study of random instances of the problem.
The negative result for deterministic algorithms holds even in the unit-profit model studied in this paper. 
This impossibility result is derived by combining one big item of size~$B$, potentially followed by $B / \epsilon$ small items of size~$\epsilon$, for a fixed $\epsilon > 0$.
After the first item, the adversary either stops the sequence there if the online algorithm did not accept the big item, or requests the small items. The competitive ratio can be made arbitrarily large by setting $\epsilon$ small enough. This simple negative result extends to multiple versions of the problem; see~\cite{DBLP:journals/mst/CyganJS16} for a recapitulation table.

\subparagraph*{Bounded item sizes}
The lower bound construction presented above can be adapted also to the case where all items have size at least $1$: in this variant, no algorithm can be better than $B$-competitive. The situation improves if we impose an additional upper bound of $m < B$ on all item sizes. For this setting (and still assuming unit profits),
Zhou, Chakrabarty and Lukose~\cite{chakrabarty2008online} gave a deterministic algorithm with competitive ratio of $O((\log B) / (1-m/B))$. In particular, their ratio becomes $O(\log B)$ when $m = o(B)$ as is the case in our paper. They also show an asymptotically matching lower bound of $\Omega(\log B)$ holding even for randomized algorithms.

While our algorithm also achieves ratios logarithmic in $B$ and is similar in structure to theirs, the results are not directly comparable. Our problem formulation assumes mixed-sign items, which introduces additional challenges in the analysis. One of them is that now the problem admits input sequences of unbounded length, leading to potentially unbounded profits. In particular, our lower bound applies even if we allow an arbitrary additive constant in the competitive ratio definition~\eqref{eq:comp_def}, and thus it is not subsumed by their lower bound of~$\Omega(\log B)$ which only holds for the (strict) competitive ratio definition without such additive constant.

The online knapsack with bounded items is a special case of online packing problems~\cite{AAP93,BuNa06} that were solved using primal-dual approaches. However, the primal-dual framework does not seem to apply here, as the presence of negative items makes the problem non-monotone.

\subparagraph*{Alternative models}
Online knapsack problems have been studied in multiple other models. In particular, a natural assumption is that profits are equal to item sizes or completely uncorrelated with them. For the former (proportional) model, deterministic algorithms still have unbounded competitive ratios, but there are $O(1)$-competitive randomized algorithms~\cite{DBLP:conf/latin/BockenhauerKKR12}.

Multiple relaxations and variations have been studied to overcome the existing lower bounds. In particular, \emph{resource augmentation}~\cite{IwaZha10,NogSar05} gives the online algorithm an advantage by allowing it to have a larger knapsack capacity than the optimal offline one. \emph{Removable items}~\cite{ABFLNE02,DBLP:journals/mst/CyganJS16,DBLP:journals/tcs/HanKM15,DBLP:conf/icalp/IwamaT02,IwaZha10} allow the algorithm to accept an item and remove it later. \emph{Multiple knapsacks}~\cite{ABFLNE02,BiPaPi20,DBLP:conf/latin/BockenhauerKKR12,BoFaLN01,DBLP:journals/mst/CyganJS16,DBLP:conf/icalp/IwamaT02} give the algorithm multiple knapsacks to fill with the goal of maximizing the total profit across all knapsacks. Finally, \emph{random-order model}~\cite{DBLP:conf/approx/Albers0L19,BaImKK07,DBLP:conf/waoa/GilibertiK21,KeRaTV18,Vaze17} assumes that the items are generated adversarially, but presented to an online algorithm in random order.


\section{Evaluations}
\label{sec:evaluations}

\subsection{Overview of the Experiments}

In this section, we evaluate \autoref{alg:det_logB} (\detlog) and compare its performance against the commonly-used greedy algorithm~\cite{DBLP:conf/mascots/AwathareSARB21}.
The code is available online~\cite{experimentcode}.
The baseline evaluation of each algorithm covers two models: one where the absolute values of the transaction amounts are restricted to $[1, B / \ln B]$ as needed for our theoretical guarantee to hold, and one where the absolute values are in the interval~$[1, B]$.
We evaluate those four variants in three settings:
\begin{enumerate}
\item
  We first consider a network consisting of a single channel (\autoref{fig:plots1}).
  We simulate the performance of the considered algorithms through that channel on transaction amounts generated uniformly at random within the ranges allowed by the considered model.
\item
  We then consider snapshot topologies of the Lightning network (\autoref{fig:plots2}).
  For each transaction, we choose the two transacting nodes uniformly at random and try to transfer the transaction using a shortest path.
  If all the channels along the routing path accept the transaction proposal, then the transaction goes through each channel, otherwise it is rejected.
  The transaction amounts are generated uniformly at random within the ranges allowed by the considered model.
\item
  In the same snapshot topology of the Lightning network, we consider a merchant scenario in which a chosen node with high capacity but limited degree receives transfers from uniformly chosen nodes~(\autoref{fig:degree-sweep}).
  Similarly, the transactions pass through the shortest path, with acceptance of each channel mandatory to the final success.
  The transaction amounts are of two types: the majority (85\%) of them have the lowest fixed \textit{amount}, and the rest have amounts chosen uniformly at random chosen in $[B_{min} / 2\ln B_{min}, B_{min} / \ln B_{min}]$ with $B_{min}$ being the capacity of the smallest edge to the merchant node. 
\end{enumerate}

\begin{figure}[t]
  \centering
  \begin{tikzpicture}[x=0.8cm,y=0.7cm, scale=0.8]
    \def\xs{1000,10000,100000}
    \def\A{2.874,3.875,4.876}      
    \def\B{2.985,3.982,4.984}      
    \def\C{2.857,3.860,4.860}      
    \def\D{2.985,3.982,4.984}      

    \definecolor{colA}{RGB}{240,173,78}
    \definecolor{colB}{RGB}{217,83,79}
    \definecolor{colC}{RGB}{92,184,92}
    \definecolor{colD}{RGB}{66,139,202}

    \draw[->] (0,0) -- (10.2,0) node[right] {$|\sigma|$};
    \draw[->] (0,0) -- (0,6) node[above] {Accepted items};

    \foreach \y/\label in {1/10, 2/100, 3/1K, 4/10K, 5/100K} {
      \draw[gray!40, densely dotted] (0,\y) -- (10,\y);
      \node[anchor=east, font=\scriptsize] at (-0.15,\y) {\label};
    }

    \foreach[count=\i] \x in \xs {
      \pgfmathsetmacro\pos{1.5 + (\i-1)*3}
      \node[font=\small] at (\pos,-0.3) {\x};
    }

    \foreach[count=\i] \d in \D {
      \pgfmathsetmacro\X{1.5 + (\i-1)*3}
      \fill[colD!85] (\X-0.75,0) rectangle (\X-0.35,\d);
    }
    \foreach[count=\i] \b in \B {
      \pgfmathsetmacro\X{1.5 + (\i-1)*3}
      \fill[colB!85] (\X-0.25,0) rectangle (\X+0.15,\b);
    }
    \foreach[count=\i] \a in \A {
      \pgfmathsetmacro\X{1.5 + (\i-1)*3}
      \fill[colA!85] (\X+0.75,0) rectangle (\X+1.15,\a);
    }
    \foreach[count=\i] \c in \C {
      \pgfmathsetmacro\X{1.5 + (\i-1)*3}
      \fill[colC!85] (\X+0.25,0) rectangle (\X+0.65,\c);
    }

    \begin{scope}[shift={(0,7.5)}]
      \draw[rounded corners=2pt, fill=white, draw=gray!60] (-0.1,-0.15) rectangle (10.1,1.3);
      \fill[colD!85] (0.2,0.75) rectangle +(0.25,0.25);
      \node[anchor=west, font=\scriptsize] at (0.55,0.325) {\detlog $[1,B]$};
      \fill[colB!85] (5.1,0.75) rectangle +(0.25,0.25);
      \node[anchor=west, font=\scriptsize] at (5.45,0.875) {\GREEDY $[1,B/\ln B]$};
      \fill[colC!85] (0.2,0.2) rectangle +(0.25,0.25);
      \node[anchor=west, font=\scriptsize] at (5.45,0.325) {\GREEDY $[1,B]$};
      \fill[colA!85] (5.1,0.2) rectangle +(0.25,0.25);
      \node[anchor=west, font=\scriptsize] at (0.55,0.875) {\detlog $[1,B/\ln B]$};
    \end{scope}

  \end{tikzpicture}
  \caption{
    Number of accepted items by four algorithms depending on the number of transactions, looking at one isolated channel.
  }
  \label{fig:plots1}
\end{figure}

\begin{figure}[t]
  \centering
  \begin{tikzpicture}[x=0.8cm,y=0.7cm, scale=0.8]
    \def\xs{1000,10000,100000}
    \def\A{2.961,3.859,4.788}      
    \def\B{2.998,3.994,4.986}      
    \def\C{2.743,3.782,4.769}      
    \def\D{2.991,3.993,4.986}      

    \definecolor{colA}{RGB}{240,173,78}
    \definecolor{colB}{RGB}{217,83,79}
    \definecolor{colC}{RGB}{92,184,92}
    \definecolor{colD}{RGB}{66,139,202}

    \draw[->] (0,0) -- (10.2,0) node[right] {$|\sigma|$};
    \draw[->] (0,0) -- (0,6) node[above] {Accepted items};

    \foreach \y/\label in {1/10, 2/100, 3/1K, 4/10K, 5/100K} {
      \draw[gray!40, densely dotted] (0,\y) -- (10,\y);
      \node[anchor=east, font=\scriptsize] at (-0.15,\y) {\label};
    }

    \foreach[count=\i] \x in \xs {
      \pgfmathsetmacro\pos{1.5 + (\i-1)*3}
      \node[font=\small] at (\pos,-0.3) {\x};
    }

    \foreach[count=\i] \d in \D {
      \pgfmathsetmacro\X{1.5 + (\i-1)*3}
      \fill[colD!85] (\X-0.75,0) rectangle (\X-0.35,\d);
    }
    \foreach[count=\i] \b in \B {
      \pgfmathsetmacro\X{1.5 + (\i-1)*3}
      \fill[colB!85] (\X-0.25,0) rectangle (\X+0.15,\b);
    }
    \foreach[count=\i] \c in \C {
      \pgfmathsetmacro\X{1.5 + (\i-1)*3}
      \fill[colC!85] (\X+0.25,0) rectangle (\X+0.65,\c);
    }
    \foreach[count=\i] \a in \A {
      \pgfmathsetmacro\X{1.5 + (\i-1)*3}
      \fill[colA!85] (\X+0.75,0) rectangle (\X+1.15,\a);
    }

    \begin{scope}[shift={(0,7.5)}]
      \draw[rounded corners=2pt, fill=white, draw=gray!60] (-0.1,-0.15) rectangle (10.1,1.3);
      \fill[colD!85] (0.2,0.75) rectangle +(0.25,0.25);
      \node[anchor=west, font=\scriptsize] at (0.55,0.325) {\detlog $[1,B]$};
      \fill[colB!85] (5.1,0.75) rectangle +(0.25,0.25);
      \node[anchor=west, font=\scriptsize] at (5.45,0.875) {\GREEDY $[1,B/\ln B]$};
      \fill[colC!85] (0.2,0.2) rectangle +(0.25,0.25);
      \node[anchor=west, font=\scriptsize] at (5.45,0.325) {\GREEDY $[1,B]$};
      \fill[colA!85] (5.1,0.2) rectangle +(0.25,0.25);
      \node[anchor=west, font=\scriptsize] at (0.55,0.875) {\detlog $[1,B/\ln B]$};
    \end{scope}

  \end{tikzpicture}
  \caption{
    Number of accepted items of four algorithms in a topology of the Lightning Network, depending on the number of transactions.
  }
  \label{fig:plots2}
\end{figure}

\begin{figure}[t]
  \centering
  \begin{tikzpicture}[x=0.8cm,y=0.7cm, scale=0.8]
    \def\xs{3,8,331}
    \def\A{3.459, 4.304, 4.362}   
    \def\B{4.306, 4.417, 4.354}   

    \definecolor{colA}{RGB}{217,83,79}
    \definecolor{colB}{RGB}{66,139,202}

    \draw[->] (0,0) -- (7.5,0) node[right] {Sink degree};
    \draw[->] (0,0) -- (0,6) node[above] {Accepted items};

    \foreach \y/\label in {1/10, 2/100, 3/1K, 4/10K, 5/100K} {
      \draw[gray!40, densely dotted] (0,\y) -- (7,\y);
      \node[anchor=east, font=\scriptsize] at (-0.15,\y) {\label};
    }

    \foreach[count=\i] \x in \xs {
      \pgfmathsetmacro\pos{1.0 + (\i-1)*2.5}
      \node[font=\small] at (\pos,-0.3) {\x};
    }

    \foreach[count=\i] \a in \A {
      \pgfmathsetmacro\X{1.0 + (\i-1)*2.5}
      \fill[colA!85] (\X+0.05,0) rectangle (\X+0.45,\a);
    }
    \foreach[count=\i] \b in \B {
      \pgfmathsetmacro\X{1.0 + (\i-1)*2.5}
      \fill[colB!85] (\X-0.45,0) rectangle (\X-0.05,\b);
    }

    \begin{scope}[shift={(0,7.5)}]
      \draw[rounded corners=2pt, fill=white, draw=gray!60] (-0.1,-0.15) rectangle (7.1,0.85);
      \fill[colB!85] (0.2,0.23) rectangle +(0.25,0.25);
      \node[anchor=west, font=\scriptsize] at (0.55,0.37) {\detlog};
      \fill[colA!85] (3.6,0.23) rectangle +(0.25,0.25);
      \node[anchor=west, font=\scriptsize] at (3.95,0.37) {\GREEDY};
    \end{scope}

  \end{tikzpicture}
  \caption{
    Number of accepted items of \GREEDY and \detlog depending on the degree of the sink node,
    for $|\sigma| = 30{,}000$ transactions generated by the merchant-bimodal distribution.
  }
  \label{fig:degree-sweep}
\end{figure}

\subsection{Methodology}
We simulate value transfers in graphs where nodes represent users and edges represent payment channels.
Any node can request a transfer to any other node.
The initial balance is always 0.
Our model is implemented on top of the Python NetworkX kit.

\subparagraph*{Baseline experiment}
For each request, two nodes are chosen uniformly at random as the source and the destination.
Then, a request size is chosen uniformly at random in a range $[1,B]$ or $[1,B/\ln B]$ depending on the considered model.
Here, $B$ is an upper bound dictated by the hash time locked contract (HTLC) request size policy and states of the channels adjacent to the source (since it knows the current balance of its own channels), and the HTLC request size policy and channel capacities of the destination (since the source can only know the predefined properties of third party channels). The minimum of these values is used.
A~path containing the least number of edges is searched between the source and the destination.
Only channels whose capacity and HTLC policy do not limit the tested request size are taken into consideration for computing the routing path.
However, this does not account for their state, as the sending party cannot know them.
If a path is found, we test it with the acceptance policy under evaluation: \GREEDY or \detlog.
If no such path exists, we try another transaction amount.
If all channels on the path accept the transfer, it is counted as successful.
\subparagraph*{Merchant scenario}
The scenario (\autoref{fig:degree-sweep}) simulates a simplified real life model, where a merchant node would be the destination of multiple transactions, most of them small. This node would have a degree (the number of adjacent open channels) that would put it between most individual users of the network and the central nodes of the exchanges and other infrastructure providers.
For each tested degree, the node with adjacency smaller or equal to that and the biggest single-edge capacity (the capacity of the biggest adjacent channel). Then, for each request, a random sender node would be uniformly chosen with the request size being: \textit{the small amount} with 15\% chance, fixed at 1000 $sat$ and the remaining 85\% taken uniformly at random from $[B_{min} / 2\ln B_{min}, B_{min} / \ln B_{min}]$ with $B_{min}$ being the capacity of the smallest edge to the merchant node. 
The routing works as in the \textit{Baseline}. Likewise, for the item to be accepted, it needs to pass through all the intermediary channels without violating their policy or capacity limits.

\subsection{Curating the Dataset}
For the actual Lightning Network topology, we use the Lightning Network Gossip dataset~\cite{lngossip}.
It is a collection of gossip messages (announcements of channel creation and updates) collected by the c-lightning observer node between 2018 and 2023.
Using the provided time machine tool and the newest dataset provided at the time of writing (gossip-20230924), we recreate a~view of the network from 23.09.2023.
For it to be usable in our simulations, we augment the edge (channel) data with the capacity information that is missing in the original dataset.
To this end, we use Mempool's public REST API \cite{mempool}.
Finally, to remove the snapshot's peculiarities like directed and duplicated edges, we removed multi-edges (channels with different SCID connecting the same pair of nodes) and merged double-edges having the same SCID, but present twice in the dataset, sometimes with different properties like HTLC policy (dictating maximum and minimum transaction size that could be routed through the node).
In such a case, the maximum of HTLC minimums and the minimum of HTLC maximums are taken, potentially narrowing the channel.
In the end, 6673 edges remained out of the original 7492 edges: we eliminated 23 multi-edges and 796 double-edges (only 131 out of them were not direct copies of themselves).

\subsection{Interpretation of the Results}

Our baseline results show that \detlog, despite its strong properties against adversarial transaction sequences, keeps the same performance as the natural, widely-used \GREEDY algorithm on common transactions, regardless of the network topology or the limitations on the transaction amounts.
\detlog therefore shows a good transaction throughput on daily transactions while ensuring a minimal throughput in any case.

The merchant scenario presents an advantageous case for \detlog. The directed nature of the transaction stream presents a more challenging case for balance management, as the balance cannot be restored by pure chance of getting a request in the opposite direction. In these conditions, the more careful \detlog gains advantage over \GREEDY. At the same time and with the current merchant selection, for the same number of requests, this effect is more pronounced when the degree of the receiving node, and coincidentally the combined capacity, is lower. Under an identical network topology, the higher number of approach channels with their additional capacity reduces the chance of failure due to their saturation.


\section{Conclusions}
This paper introduced and studied a fundamental model in decentralized finance to make cryptocurrencies scalable and able to process daily payments.
The problem translates into an~online variant of the knapsack problem in which items may have positive or negative sizes, subject to symmetric capacity constraints.
We derived a deterministic online algorithm with asymptotically optimal competitive ratio for this problem, and showed that the ratio is tight even for randomized algorithms.
Our simulations suggest that our algorithm, designed for worst-case scenarios, performs very well on average transaction sequences.
An important direction for future work is to extend the model from a single channel to the full network setting of payment channel systems, where transactions must be routed across multiple channels with independently evolving states.

\bibliography{references}

\end{document}